\documentclass[runningheads]{llncs}

\usepackage[T1]{fontenc}
\usepackage{graphicx}
\usepackage{amsmath}
\usepackage{amssymb}
\usepackage{amsthm}
\usepackage{caption}
\usepackage{subcaption}

\begin{document}
\title{Adding a Tail in Classes of Perfect Graphs\thanks{
Research at the University of Ioannina supported by the Hellenic Foundation for Research and Innovation (H.F.R.I.) under the ``First Call for H.F.R.I. Research Projects to support Faculty members and Researchers and the procurement of high-cost research equipment grant'', Project FANTA (eFficient Algorithms for NeTwork Analysis), number HFRI-FM17-431.
}
}
%
%
\author{Anna Mpanti\inst{1,2}
\and
Stavros D. Nikolopoulos\inst{1,3}
\and
Leonidas Palios\inst{1,4}
}
\authorrunning{A. Mpanti et al.}
%
\institute{Dept. of Computer Science and Engineering, University of Ioannina, Greece
\and
ampanti@cs.uoi.gr
\and
0000-0001-6684-8459; \ stavros@cs.uoi.gr
\and
0000-0001-8630-3835; \ palios@cs.uoi.gr
}
\maketitle              
\begin{abstract}
Consider a graph~$G$ which belongs to a graph class~${\cal C}$. We are interested in connecting a node~$w \not\in V(G)$ to $G$ by a single edge~$u w$ where $u \in V(G)$; we call such an edge a \emph{tail}. As the graph resulting from $G$ after the addition of the tail, denoted $G+uw$, need not belong to the class~${\cal C}$, we want to compute a minimum ${\cal C}$-completion of $G+w$, i.e., the minimum number of non-edges (excluding the tail~$u w$) to be added to $G+uw$ so that the resulting graph belongs to ${\cal C}$.

In this paper, we study
this problem
for the classes of split, quasi-thre\-shold, threshold, and $P_4$-sparse graphs and we present linear-time algorithms by exploiting the structure of split graphs and the tree representation of quasi-threshold, threshold, and $P_4$-sparse graphs. 



\keywords{edge addition  \and completion \and split graph \and quasi-threshold graph \and threshold graph \and $P_4$-sparse graph}
\end{abstract}

\section{Introduction}

Given a graph~$G$, an edge connecting a vertex~$w \not \in V(G)$ to a vertex~$u$ of $G$ is a \emph{tail} added to $G$; let us denote the resulting graph as $G + uw$. If $G$ belongs to a class~${\cal C}$ of graphs, this may not hold for the graph~$G + uw$. Hence, we are interested in computing a minimum ${\cal C}$-completion of $G + uw$, i.e., the minimum number of non-edges (excluding the tail~$u w$) to be added to $G + uw$ so that the resulting graph belongs to ${\cal C}$; such non-edges are called \emph{fill edges}. The above problem is an instance of the more general (${\cal C},+k$)-MinEdgeAddition problem \cite{NP05} in which we add $k$ given non-edges in a graph belonging to a class~${\cal C}$ and we want to compute a minimum ${\cal C}$-completion of the resulting graph.

Computing a minimum completion of an arbitrary graph into a specific graph class is an important and well studied problem with applications in areas involving graph modeling with missing edges due to lacking data, e.g., molecular biology and numerical algebra \cite{GGKS95,NSS01,R72}. Unfortunately, minimum completions into many interesting graph classes, such as split graphs, chordal graphs and cographs, are NP-hard to compute \cite{BBD06,EC88,KF79,Y81}. This led the researchers towards the computation of minimal completions, the solution of problems with restricted input, and approximation or parameterized algorithms.

A related field is that of the dynamic recognition (or on-line maintenance) problem on graphs: a series of requests for the addition or the deletion of an edge or a vertex (potentially incident on a number of edges) are submitted and each is executed only if the resulting graph remains in the same class of graphs. Several authors have studied this problem for different classes of graphs and have given algorithms supporting some or all the above operations; we mention the edges-only fully dynamic algorithm of Ibarra \cite{I08} for chordal and split graphs, and the fully dynamic algorithms of Hell et al{.} \cite{HSS02} for proper interval graphs, of Shamir and Sharan \cite{SS04} for cographs, of Heggernes and Mancini for split graphs \cite{HM09}, and of Nikolopoulos et al{.} for $P_4$-sparse graphs \cite{NPP12}.

In this paper, we exploit the structure of split graphs and the tree representation of quasi-threshold, threshold, and $P_4$-sparse graphs in order to present algorithms for computing a minimum completion of a given graph~$G$ in each of these classes to which we have added a tail. Given the ($K,S$)-partition of a given split graph or the tree representation of a given quasi-threshold, threshold, or $P_4$-sparse graph, our algorithms run in optimal $O(n)$ time where $n$ is the number of vertices of $G$. These algorithms are a first step towards the solution of the (${\cal C},+1$)-MinEdgeAddition problem \cite{NP05} for each of these four classes~${\cal C}$ of graphs.

\section{Theoretical Framework}
We consider finite undirected graphs with no loops or multiple edges. For a graph~$G$, we denote by $V(G)$ and $E(G)$ the vertex set and edge set of $G$, respectively.
Let $S$ be a subset of the vertex set~$V(G)$ of a graph~$G$. Then, the subgraph of $G$ induced by $S$ is denoted by $G[S]$.
The {\it neighborhood\/}~$N_G(x)$ of a vertex~$x$ of the graph~$G$ is the set of all the vertices of $G$ which are adjacent to $x$. The {\it closed neighborhood\/} of $x$ is defined as $N_G[x] := N_G(x) \cup \{x\}$.
The {\it degree\/} of a vertex $x$ in $G$, denoted $deg(x)$, is
the number of vertices adjacent to $x$ in $G$; thus, $deg(x) = |N_G(x)|$.
A vertex of a graph is \emph{universal} if it is adjacent to all other vertices of the graph. We extend this notion to a subset of the vertices of a graph~$G$ and we say that a vertex is \emph{universal in a set}~$S \subseteq V(G)$, if it is universal in the induced subgraph~$G[S]$.
Finally, $C_k$ ($P_k$ resp.) denotes the chordless cycle (chordless path resp.) on $k$ vertices; in each $P_4$, the unique edge incident on its first or last vertex is called a \emph{wing}.

\section {Split Graphs}
The split graphs are of wide theoretical interest and have been the focus of many research papers. An undirected graph $G$ is \emph{split} if its vertex set~$V(G)$ admits a partition into a clique~$K$ and an independent set~$S$ \cite{Go2}; the partition into $K, S$ can be computed in time proportional to the size of the graph.
It also holds that a graph is split if and only if it contains no induced $C_4$, $C_5$, or $2K_2$.


\begin{lemma} \label{lemma:forms}
Let $G$ be a split graph with vertex partition into a clique~$K$ and an independent set~$S$, $u$ a vertex of $G$, $uw$ a tail, and $K_s = \{ x \in K \,|\, N_G(x) \cap S \ne \emptyset \}$. 
Then, there exists a split-completion for the graph~$G+uw$ in which the number of fill edges needed
is
$0$ if $u \in K$ and
$|K_s|- deg_G(u)$ if $u \in S$.
\end{lemma}
\begin{proof}
If $u \in K$, no fill edge (in addition to $uw$) is needed, which is optimal, since $G+uw$ is a split graph with clique~$K$ and independent set~$S \cup \{w\}$.

Now consider that $u \in S$. A split completion of $G+uw$ can be obtained by connecting $u$ to all its non-neighbors in $K_s$; the resulting graph is split with clique~$K_s \cup \{u\}$ and independent set~$S \cup (K - K_s) \cup \{w\}$.
To prove its optimality, suppose for contradiction that there existed a split completion of $G+uw$ that uses fewer than $|K_s| - deg_G(u)$ fill edges. Then, there would exist a vertex $a \in K_s \setminus N_G(u)$ which is not incident on any fill edge. If there existed one more vertex $b \in K_s \setminus (N_G(u) \cup \{a\})$ not incident on any fill edge as well, then the edges $ab$ and $uw$ would form a $2 K_2$, a contradiction. Then, all the fill edges would be incident on the vertices in $K_s \setminus (N_G(u) \cup \{a\})$. But then, if $z$ is a neighbor of $a$ in $S$, the edges $a z$ and $uw$ would form a $2 K_2$, a contradiction.
\end{proof}
Since the vertices in the clique~$K$ are all pairwise adjacent, we note that $|K \setminus K_s| \le 1$.
The lemma directly implies that given the set~$K_s$, the minimum number of fill edges can be computed in $O(|V(G)|)$ time otherwise the time complexity is $O(|V(G)| + |E(G)|)$.

\begin{figure}[t!]
    \centering
    \includegraphics[width=0.6\textwidth]{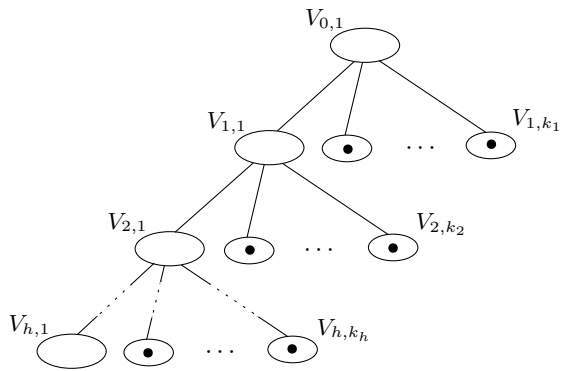}
    \caption{The structure of the tree representation of a threshold graph \cite{Ni}
    } \label{fig:thres}
\end{figure}

\section {Threshold and Quasi-threshold Graphs}

\noindent
\textbf{Threshold Graphs}. A well-known subclass of perfect graphs called threshold graphs are those whose independent vertex set subsets can be distinguished by using a single linear inequality. 
A graph~$G$ is \emph{threshold} if there exists a threshold assignment $[\alpha, t]$ consisting of a labeling $\alpha$ of the vertices by non-negative integers and an integer threshold~$t$ such that: a set~$S \subseteq V(G)$ is independent if and only if $\alpha(v_1) + \alpha(v_2) + \cdots + \alpha(v_p) \leq t$ where $v_i\in S, 1 \leq i \leq p$. Chvátal and Hammer \cite{ChHa} first proposed threshold graphs in 1973 and have proved that the threshold graphs are precisely the graphs that contain no induced $C_4$, $P_4$, or $2K_2$. 

Nikolopoulos \cite{Ni} proved that every threshold graph admits a unique rooted tree representation as shown in Figure~\ref{fig:thres}: each tree node is associated with a vertex set $V_{i,j}$ (these sets partition the vertex set of the graph) with each $V_{i,1}$ inducing a clique and each of the remaining sets containing a single vertex (note that the tree nodes associated with these singleton sets have no descendants) and the vertices in the union of the sets associated with the nodes on a path of tree edges from a tree node to any one of its descendants induce a clique. Thus, the vertices in $V_{k,1}$ are adjacent to all the vertices in $\bigcup_{i>k} V_{i,j}$, the vertices in $\bigcup_i \bigcup_{j\ge 2} V_{i,j}$ form an independent set, and the vertices in $\bigcup_i V_{i,1}$ induce a clique.

Let $G$ be a threshold graph and consider adding a tail~$uw$ to $G$ where $u \in V(G)$. Then we can show the following lemma.

\begin{figure}[t!]
\centering
\includegraphics[width=0.55\textwidth]{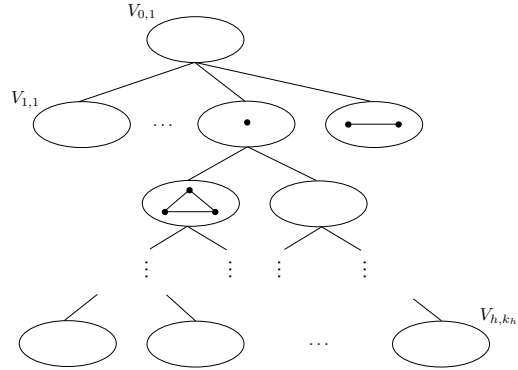}
\caption{The tree representation of a quasi-threshold graph.}
\label{fig:qt}
\end{figure}

\begin{lemma} \label{lemma:forms}
Let $G$ be a threshold graph and let its tree representation~$T_G$ consist of nodes associated with sets~$V_{i,j}$ where $0 \le i \le h$ and $1 \le j \le k_i$ (Figure~\ref{fig:thres}). Consider the addition of a tail~$u w$ where $u \in V(G)$. Then, there exists a minimum threshold completion of the graph~$G+uw$ using $f$ fill edges (excluding $uw$) where: 
\begin{enumerate}
\item [(i)] If $u \in V_{i,1}$, then $f \ =\  \min_{0\leq \ell \leq i} \left\{ \left( \sum_{s=\ell+1}^{i} \sum_{t=2}^{k_{s}} |V_{s,t}| \right) + \sum_{s=0}^{\ell} |V_{s,1}| \right\}$;

\item [(ii)] If $u  \in V_{i,j}$ where $2 \leq j \leq k_i$, then $f = \min\{f_1, f_2\}$ where\\
$f_1 \ =\  \sum_{r=i}^{h} |V_{r,1}| + \min_{i\leq \ell \leq h} \left\{ \left( \sum_{s=\ell+1}^{h} \sum_{t=2}^{k_{s}} |V_{s,t}| \right) + \sum_{s=0}^{\ell} |V_{s,1}| \right\}$ \ and\\
$f_2 \ =\  \sum_{r=i}^{h} |V_{r,1}| + \min_{0\leq \ell \leq i-1} \left\{ \left( \sum_{s=\ell+1}^{h} \sum_{t=2}^{k_{s}} |V_{s,t}| \right) + \sum_{s=0}^{\ell} |V_{s,1}| \right\} - 1$.


\end{enumerate}
\end{lemma}

\bigskip
\noindent
\textbf{Quasi-threshold Graphs}. A graph $G$ is called \emph{quasi-threshold}, or QT-graph for short, if $G$ contains no induced
$C_4$ or $P_4$ \cite{Go,MaWaWu,Wo1,Wo2}. The class of quasi-threshold-graphs is a subclass of the class of cographs and properly contains the class of threshold graphs \cite{CoLeBu,CoPeSt,Go2,Ve}. Nikolopoulos and Papadopoulos \cite{NiPa} have shown, among other properties, a unique rooted tree representation of QT-graphs which is a generalization of the tree representation of threshold graphs (Figure~\ref{fig:qt}): the tree nodes are associated with disjoint vertex subsets each inducing a clique and the vertex sets associated with the tree nodes in a path from a tree node to any of its descendants induce a clique. It has been proven that a graph is QT-graph if and only if it admits such a tree representation \cite{KaNi,Ni2}. Then, by generalizing the approach for threshold graphs, we can show the following lemma.

\begin{lemma} \label{lemma:forms}
Let $G=(V,E)$ be a QT-graph, and let $T_G$ be its tree representation. Consider the addition of a tail~$u w$ incident on a node~$u$ of $G$ and suppose without loss of generality that $u \in V_{i,1}$ and that the vertex sets associated with the tree nodes on the path from the root of $T_G$ to the node associated with $V_{i,1}$ are in order $V_{0,1}, V_{1,1}, \ldots, V_{i,1}$. Then, there exists a minimum QT completion of the graph~$G+uw$, and the minimum number of fill edges needed (excluding the tail~$u w$) is
$\min_{0\leq \ell \leq i} \left\{ \left( \sum_{s=\ell+1}^{i} \sum_{t=2}^{k_{s}} |V_{s,t}| \right) + \sum_{s=0}^{\ell} |V_{s,1}| \right\}$

\end{lemma}


\section{$P_4$-sparse Graphs}

The $P_4$-sparse graphs are defined as the graphs for which every set of $5$ vertices induces at most one $P_4$ \cite{H85} (Figure~\ref{fig:forbidden} depicts the $7$ forbidden subgraphs for the class of $P_4$-sparse graphs). The $P_4$-sparse graphs are perfect and also perfectly orderable \cite{H85}, and properly contain many graph classes, such as, the cographs, the $P_4$-reducible graphs, etc. (see \cite{BLS99,JO92a,JO92b}). They have received considerable attention in recent years and find applications in applied mathematics and computer science (e.g., communications, transportation, clustering, scheduling, computational semantics) in problems that deal with graphs featuring “local density” properties.

\begin{figure}[t!]
\centering
\includegraphics[width=4.8in]{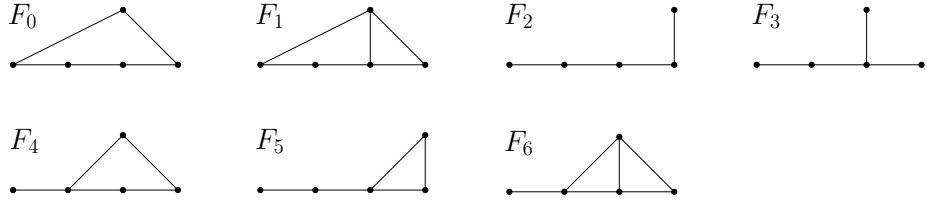}
\caption{The forbidden subgraphs of the class of $P_4$-sparse graphs 
\cite{JO92b}.}\label{fig:forbidden}
\end{figure}

For a $P_4$-sparse graph either the graph or its  complement is disconnected with the connected components inducing $P_4$-sparse graphs, or induces a spider. A graph~$H$ is called a \emph{spider} if its vertex set~$V(H)$ admits a partition into sets $S, K, R$ such that:
\begin{itemize}
\item
the set~$S$ is an independent set, the set~$K$ is a clique, and $|S| = |K| \ge 2$;
\item
every vertex in $R$ is adjacent to every vertex in $K$ and to no vertex in $S$;
\item
there exists a bijection $f: S \to K$ such that for each vertex~$s \in S$ either $N_G(s) \cap K = \{f(s)\}$ or $N_G(s) \cap K = K - \{f(s)\}$; in the former case, the spider is \emph{thin}, in the latter it is \emph{thick} (see Figure~\ref{fig:spiders}).
\end{itemize}
Note that for $|S| = |K| = 2$, the spider is simultaneously thin and thick. To avoid ambiguity, in the following, for thick spiders we assume that $|K| \ge 3$.


\begin{figure}[t]
\centering
\begin{minipage}{.7\textwidth}
  \centering
  \includegraphics[width=.99\linewidth]{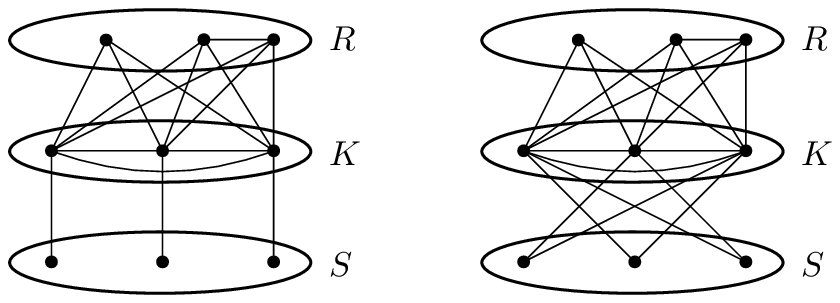}
  \captionof{figure}{(left)~A thin spider; (right)~a thick spider.}
  \label{fig:spiders}
\end{minipage}%
\qquad
\begin{minipage}{.21\textwidth}
  \centering
  \includegraphics[width=.99\linewidth]{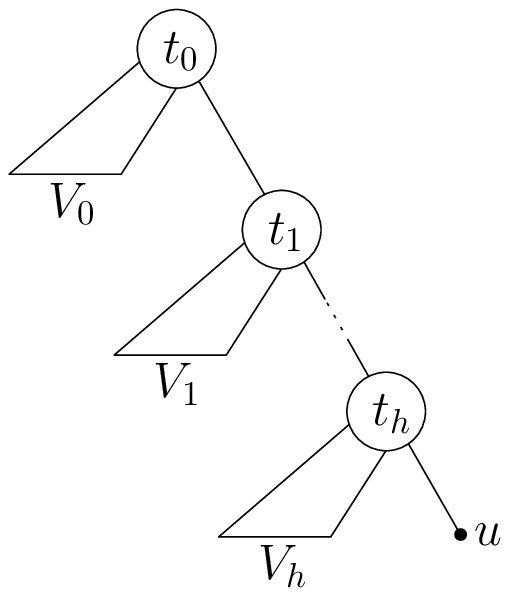}
  \captionof{figure}{}
  \label{fig:labelV}
\end{minipage}
\end{figure}

In \cite{JO92b}, Jamison and Olariu showed that each $P_4$-sparse graph~$G$ admits a unique tree representation, up to isomorphism, called the \emph{$P_4$-sparse tree} $T(G)$ of $G$, which is a rooted tree such that:
\begin{itemize}
\item[(i)]
each internal node of $T(G)$ has at least $2$ children provided that $|V(G)| \ge 2$;
\item[(ii)]
the internal nodes are labelled by one of $0$, $1$, or $2$ (\emph{$0$-, $1$-, $2$-nodes, resp.}) and the parent-node of each $0$- or $1$-node~$t$ has a different label than $t$;
\item[(iii)]
the leaves of the $P_4$-sparse tree are in a $1$-to-$1$ correspondence with the vertices of $G$; if the least common ancestor of the leaves corresponding to two vertices $v_i, v_j$ of $G$ is a $0$-node ($1$-node, resp.) then the vertices $v_i, v_j$ are non-adjacent (adjacent, resp.) in $G$, whereas the vertices corresponding to the leaves of a subtree rooted at a $2$-node induce a spider.
\end{itemize}

The structure of the $P_4$-sparse tree implies the following lemma.

\begin{lemma} \label{lemma:p4sparse_tree}
Let $G$ be a $P_4$-sparse graph and let $H = (S,K,R)$ be a \emph{thin} spider of $G$. Moreover, let $s \in S$ and $k \in K$ be vertices that are adjacent in the spider.
\par\smallskip\noindent
\textbf{\ P1}.
Every vertex of the spider is adjacent to all vertices in $N_G(s) \setminus \{k\}$.
\par\smallskip\noindent
\textbf{\ P2}.
Every vertex~$z \in K \setminus \{k\}$ is adjacent to all vertices in $N_G(k) \setminus \{s,z\}$.
\end{lemma}
\par\smallskip


Let $G$ be a given graph to which we want to add the tail~$uw$ with $u \in V(G)$. 
Let $t_0 t_1 \cdots t_h u$ be the path from the root~$t_0$ of the $P_4$-sparse tree~$T_G$ of $G$ to the leaf associated with $u$. Moreover, let $V_i$ ($0 \le i < h$) be the set of vertices associated with the leaves of the subtrees rooted at the children of $t_i$ except for $t_{i+1}$ and $V_h$ be the set of vertices associated with the leaves of the subtrees rooted at the children of $t_h$ except for the leaf associated with $u$ (see Figure~\ref{fig:labelV}). The sets~$V_0, V_1, \ldots, V_h$ form a partition of $V(G) \setminus \{u\}$.

We 
show that there always exists a minimum $P_4$-sparse completion of the graph $G+uw$ exhibiting one of a small number of different formations for $u,w$.


\begin{lemma} \label{lemma:forms}
Let $G$ be a $P_4$-sparse graph and $T_G$ be its $P_4$-sparse tree. Consider the addition of a tail~$u w$ incident on a node~$u$ of $G$. Then, there exists a minimum $P_4$-sparse completion~$G'$ of the graph~$G+uw$ such that for the $P_4$-sparse tree~$T_{G'}$ of $G'$, one of the following three cases holds:
\begin{enumerate}
\item  
The nodes $u, w$ in $T_{G'}$ have the same parent-node which is a $2$-node corresponding to a thin spider 
$(S,K,R)$ with $u \in K$ and $w \in S$.
\item  
The $P_4$-sparse tree~$T_{G'}$ results from $T_G$ by replacing the leaf for $u$ by the $3$-treenode Formation~$1$ shown in Figure~\ref{fig:formations}(left).
%
\item  
The $P_4$-sparse tree~$T_{G'}$ results from $T_G$ by removing the leaf for $u$ and replacing an $1$- or a $2$-node~$t$ in the path from the root of $T_G$ to the leaf for $u$ by the $5$-treenode Formation~$2$ in Figure~\ref{fig:formations}(right).
\end{enumerate}
\end{lemma}

\begin{figure}[t!]
\centering
\includegraphics[width=2in]{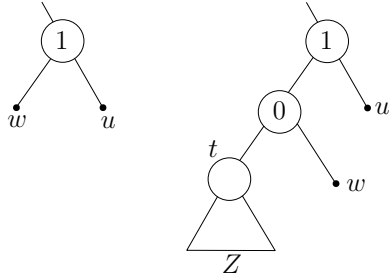}
\caption{(left)~Formation~$1$; (right)~Formation~$2$ where $t$ is a 1- or a 2-node.
Formation~$1$ is a special case of Formation~$2$ when $Z = \emptyset$.}
\label{fig:formations}
\end{figure}

\begin{proof}
Let $G_{OPT}$ be a minimum $P_4$-sparse completion of the graph~$G + uw$ and let $T_{OPT}$ be its $P_4$-sparse tree. We consider the following cases:

\smallskip
\noindent
A. \textit{The leaves associated with $u, w$ in $T_{OPT}$ do not have the same parent-node}:
\ Let $T_R$ be the $P_4$-sparse tree obtained from $T_{OPT}$ by using Formation~$2$ just above the least common ancestor~$t$ of $w$ and $u$ in $T_{OPT}$ (Figure~\ref{fig:to_form_2}); let $G_R$ be the $P_4$-sparse graph corresponding to the tree~$T_R$. Then, $G_R$ uses no more fill edges than $T_{OPT}$. To see this, let $t'$ be the child of $t$ that is an ancestor of the leaf for $u$ (note that $t'$ may coincide with the leaf for $u$). Since $u, w$ are adjacent in $G_{OPT}$, $t$ is a $1$- or a $2$-node. In either case, $w$ is adjacent to all vertices corresponding to the leaves of the subtree of $T_{OPT}$ rooted at $t'$ and all these edges, except for the tail~$uw$, are fill edges. If $t$ is a $1$-node, then $u$ is adjacent to all vertices in $X$ (Figure~\ref{fig:to_form_2}) and thus $G_R$ uses no more fill edges. If $t$ is a $2$-node then $u$ is adjacent to all the vertices in the clique~$K_X$ of the corresponding spider (which includes $w$). Moreover, because $w \in K_X$, $w$ is adjacent to all the vertices in $K_X \setminus \{w\}$ and to at least $1$ vertex in the independent set
for a total of $|K_X|$ fill edges; these fill edges can be used to connect $u$ to 
the vertices in the in\-de\-pen\-dent set of the spider and thus $G_R$ uses no more fill edges in this case too.

Now, in the $P_4$-sparse tree~$T_R$ in Figure~\ref{fig:to_form_2}(right), let $A = V(G) \setminus (Z \cup \{u\})$. 
Recall that in the $P_4$-sparse tree~$T_G$ of $G$, the path from the root~$t_0$ to $u$ is $t_0 t_1 \cdots t_h u$ and $V_i$ ($0 \le i \le h$) is the set of vertices associated with the leaves of the subtrees rooted at the children of $t_i$ except for $t_{i+1}$ (where $t_{h+1}$ is the leaf associated with $u$); see Figure~\ref{fig:labelV}.
%

We first observe that the induced subgraph~$G_R[Z]$ induced by the set of vertices~$Z$ corresponding to the leaves of the subtree of $T_R$ rooted at node~$t$ coincides with the induced subgraph~$G[Z]$ (otherwise $G_{OPT}$ would include fill edges that could be removed in contradiction to its optimality); then, let $t = t_k$.
It also holds that node~$t$ in $T_R$ is a $1$- or a $2$-node, since node~$t$ was a $1$- or a $2$-node in $T_{OPT}$,  as well.
Let $A = V(G) \setminus (Z \cup \{u\})$.
Note that there is no set~$V_j$ such that $x,y \in V_j$, $x$ is a neighbor of $u$ in $G$, $x \in V_j \cap A$ and $y \in V_j \cap Z$, otherwise we can move $x$ to $Z$ along with $y$; because $y$ is in $Z$, all adjacencies from $y$ to all the vertices in $V(G) \setminus (V_j \cup \{u\})$ in $G$ are maintained and this will also hold for $x$, and the fill edge~$x w$ will be removed, a contradiction.
Similarly, there is no set~$V_j$ such that $x,y \in V_j$, $y$ is a non-neighbor of $u$ in $G$ and $x \in V_j \cap A$ and $y \in V_j \cap Z$ otherwise we can move $y$ to $A$ along with $x$. This implies that for each $i =0,1,\ldots,h$, either $V_i \subseteq A$ or $V_i \subseteq Z$, and since $t = t_k$, $V_k \subseteq Z$.

Finally, there exists no $j > k$ such that $V_j \subseteq A$. Suppose that there existed such a $V_j$ and let $j$ be the largest such index.  Then, because $t = t_k$ is a $1$- or a $2$-node and $k < j$, there would exist vertex $z \in V_k$ which would be adjacent to all vertices in $V_j$. This implies that in $T_R$, the least common ancestor of $z$ and the vertices in $V_k$ would be a $1$-node. But then, if we moved $V_j$ to $Z$ then we would have fewer fill edges, a contradiction. Therefore, the tree~$T_R$ is as described in Case~$3$ of the statement of the lemma.

\begin{figure}[t!]
\centering
\includegraphics[width=3.2in]{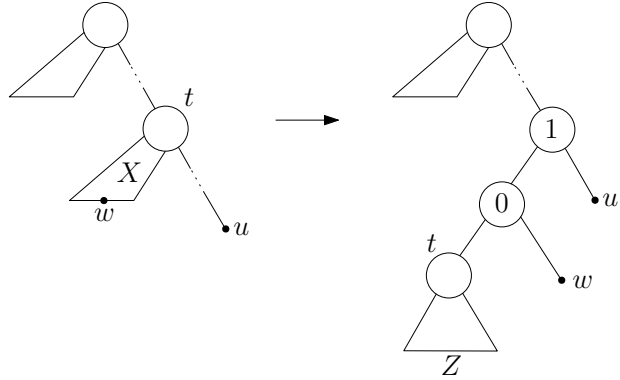}
\caption{(left)~The $P_4$-sparse tree~$T_{OPT}$ in which the leaves for $u,w$ do not have the same parent-node and have node~$t$ as their least common ancestor; (right)~The $P_4$-sparse tree~$T_R$ obtained by using Formation~$2$ just above node~$t$ which results in no more fill edges than those in $G_{OPT}$.}
\label{fig:to_form_2}
\end{figure}

\smallskip\noindent
B. \textit{The leaves associated with $u, w$ in $T_{OPT}$ have the same parent-node~$p$}: Then, since $u, w$ are adjacent, the parent-node~$p$ is either an $1$-node or a $2$-node.
\begin{itemize}
\item[(i)]
\textit{The parent-node~$p$ of $u, w$ in $T_{OPT}$ is an $1$-node}:
\ Then, the leaves associated with $u$ and $w$ are the only children of $p$ (Formation~$1$), otherwise we can use Formation~$2$ as shown in Figure~\ref{fig:temp1} which requires fewer fill edges. Then, $w$ will be adjacent to all neighbors of $u$ in $T_{OPT}$; this and the optimality of $G_{OPT}$ imply that $T_{OPT}$ results from $T_G$ by replacing the leaf for $u$ by Formation~$1$.

\item[(ii)]
\textit{The parent-node~$p$ of $u, w$ in $T_{OPT}$ is a $2$-node}:
\ Let $H = (S,K,R)$ be the corresponding spider. 
If $H$ is thick (thus $|K| \ge 3$), then 
no matter whether the tail~$u w$ is an $S$-$K$, $K$-$K$, or $R$-$K$ edge, 
the sum of degrees of $u, w$ in $H$ (excluding $uw$) is at least $|V(H)|-3 + |K| - 2$ (consider an $S$-$K$ edge). However, we would have added no more fill edges if we have made $u$ universal in $G[V(H) \setminus \{w\}]$ and then applied Formation~$2$ at the parent of the leaf for $u$ (then $Z = V(H) \setminus \{u, w\}$) using $V(H) - 2 \le V(H) + |K| -5$ fill edges.

In the same way, we show that we would have added no more fill edges if $H$ is a thin spider and the tail~$u w$ is a $K$-$K$ or $K$-$R$ edge. If $uw$ is an $S$-$K$ edge with $u \in S$ and $w \in K$, then we exchange $u$ and $w$ for the same total number of fill edges and get a thin spider with $u \in K$ and $w \in S$.
\end{itemize}
\end{proof}

\begin{figure}[t]
\centering
\includegraphics[width=2.5in]{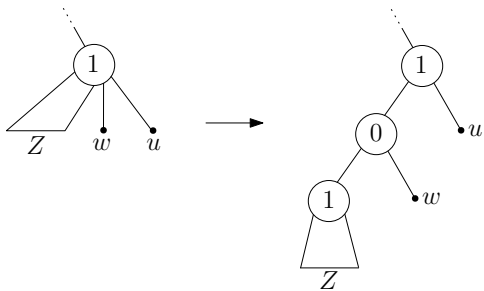}
\caption{A transformation that reduces the number of fill edges.}
\label{fig:temp1}
\end{figure}

\noindent

\subsection{Adding a Tail to a Spider}

In this section, we consider 
adding a tail~$u w$ to a spider~$H = (S_H, K_H, R_H)$ where $u \in V(H)$.
In the following two lemmas, we address the cases of a thin or a thick spider~$H$ respectively.

\begin{lemma} \label{lemma:tail_thin}
Consider the addition of a tail~$u w$ to a \emph{thin} spider~$H = (S_H, K_H$, $R_H)$ where $u$ is a vertex of $H$. Then, for the number~$f$ of fill edges (excluding the tail~$uw$) in a minimum $P_4$-sparse completion of the graph~$H+uw$, it holds:
\begin{enumerate}
\item
if $u \in S_H$, $f = |K_H|-1$ if $R_H = \emptyset$ and $f = |K_H|$ otherwise;
\item
if $u \in K_H$, $f = |K_H|-1$;
\item
if $u \in R_H$, 
then $f = \min\{\, |R_H \setminus N_H[u]|, \  |K_H| + f' \,\}$ where $f'$ is the number of fill edges (excluding $uw$) in a minimum $P_4$-sparse completion of the graph~$H[R_H]+uw$.
\end{enumerate}
\end{lemma}
\begin{proof}
1. Let $v \in K_H$ be the neighbor of $u$ in $H$. Then, we can get a $P_4$-sparse graph as follows: if $R_H = \emptyset$, we connect $u$ to all vertices in $K_H \setminus \{v\}$ (we get a thin spider with $S = (S_H \setminus \{u\}) \cup \{w\}$, $K = (K_H \setminus \{v\}) \cup \{u\}$, and $R = \{v\}$, that is, the tail~$uw$ is a wing of a $P_4$ of a thin spider), otherwise we connect $v$ to all vertices in $\{w\} \cup (S_H \setminus \{ u \})$, which makes $v$ universal in $V(H) \cup \{w\}$ and $u, w$ form a separate connected component in $G[V(G) \setminus \{v\}]$; the total number of fill edges (excluding the tail~$u w$) is precisely $|K_H|-1$ if $R_H = \emptyset$ and $K_H$ otherwise.

Moreover, this is the minimum number of fill edges (excluding $u w$) needed. 
First, we note that for each pair $k_i, s_i$ where $k_i \in K_H \setminus \{v\}$ and $s_i \in S_H \setminus \{u\}$, the vertices $v,u,w,k_i,s_i$ define an $F_5$ or an $F_3$ depending on whether the vertices $v,w$ are adjacent or not, which implies that at least $|K_H|-1$ fill edges (excluding $u w$) are needed.
Then, if there is a way of getting a $P_4$-sparse graph by adding fewer than the number of fill edges mentioned in Case~$1$ of the statement of the lemma, it has to be the case that (i)~$R_H \ne \emptyset$, (ii)~each pair $k_i, s_i$ where $k_i \in K_H \setminus \{v\}$ and $s_i \in S_H \setminus \{u\}$ is incident on exactly $1$ fill edge, and (iii)~no more fill edges exist. Let $r \in R_H$ and $k \in K_H \setminus \{v\}$. 
Then, the vertices $v, u, w, k, r$ induce a forbidden subgraph (an $F_5$ if $k$ is non-adjacent to both $u, w$, or an $F_6$ ($F_1$, resp.) if $k$ becomes adjacent to $u$ ($w$, resp.) by means of a fill edge); thus, at least $K_H$ fill edges are needed in this case.

\smallskip
2. Let $v \in S_H$ be the neighbor in $H$ of $u \in K_H$. Then, by connecting $u$ to all vertices in $S_H \setminus \{ v \}$ (which makes $u$ universal in $H$) or by connecting $w$ to all vertices in $K_H \setminus \{ u \}$ yields a $P_4$-sparse graph. Moreover, this is the minimum number of fill edges (excluding the tail~$u w$) that need to be added. Suppose, for contradiction, that we get a $P_4$-sparse graph after having added fewer than $|K_H|-1$ fill edges (excluding $u w$) to the thin spider~$H$. Then, there exists a pair of adjacent vertices~$s, k$ with $s \in S_H \setminus \{ v \}$ and $k \in K_H \setminus \{ u \}$ such that neither $s$ nor $k$ is incident on a fill edge. Then the vertices $u, v, w, s, k$ induce a forbidden subgraph~$F_5$ or $F_3$ if $w$ and $v$ are adjacent or not, respectively, a contradiction.

\smallskip
3. The term~$R_H \setminus N_H[u]$ corresponds to making $u$ universal in $H[R_H]$ in which case the resulting graph is $P_4$-sparse (it is a thin spider with $S = S_H \cup \{ w\}$, $K = K_H \cup \{ u \}$, and $R = R_H \setminus \{u\}$). The term~$|K_H| + f'$ corresponds to adding $|K_H|$ fill edges connecting $w$ to the vertices in $K_H$ and then computing a minimum $P_4$-sparse completion of the graph~$H[R_H]+uw$. Note that no minimum $P_4$-sparse completion of $H+uw$ exists with $u$ not being universal in $R_H$ and with using fewer than $|K_H|$ fill edges incident on the vertices in $S_H \cup K_H$: if there were such a minimum $P_4$-sparse completion~$H'$ of $H+uw$, then in $H'$, there would exist a non-neighbor~$r \in R_H$ and a pair of adjacent vertices $s,k$ where $s \in S_H$ and $k \in K_H$ such that neither $s$ nor $k$ would be incident on a fill edge; but then, in $H'$, the vertices $u,w,r,s,k$ induce an $F_4$ or an $F_3$ if $w,r$ have been connected by a fill edge or not, respectively, which leads to a contradiction. In turn, if $H'$ has at least $|K_H$ fill edges incident on vertices in $S_H \cup K_H$ then $H'[R_H \cup \{w\}]$ would be $P_4$-sparse using fewer than $f'$ fill edges in contradiction to the minimality of $f'$.
\end{proof}

\begin{lemma} \label{lemma:tail_thick}
Consider the addition of a tail~$uw$ to a \emph{thick} spider $H = (S_H$, $K_h$, $R_H$) where $u$ is a vertex of $H$.
Then, for the number~$f$ of fill edges (excluding the tail~$uw$) in a minimum $P_4$-sparse completion of the graph~$H+uw$, it holds:
\begin{enumerate}
\item
if $u \in S_H$,
\[
f \ =\ 
\begin{cases}
\ |K_H|-1=2 & \mbox{if $|K_H| = 3$ and $R_H = \emptyset$} \\
\ |K_H|=3 & \mbox{if $|K_H| = 3$ and $|R_H| = 1$} \\
\ |K_H|+1=4 & \mbox{if $|K_H| = 3$ and $|R_H| \ge 2$} \\
\ |K_H| & \mbox{if $|K_H| \ge 4$ and $R_H = \emptyset$} \\
\ |K_H|+1 & \mbox{if $|K_H| \ge 4$ and $|R_H| \ge 1$};
\end{cases}
\]
\item
if $u \in K_H$, $f = 1$;
\item
if $u \in R_H$, then $f = |K_H| + f'$ where $f'$ is the number of fill edges (excluding $uw$) in a minimum $P_4$-sparse completion of the graph~$H[R_H]+uw$.
\end{enumerate}
\end{lemma}

\begin{proof}
1. Let $v \in K_H$ be the non-neighbor of $u$ in $H$. Let us first consider the case $|K_H| = 3$. If $|R_H| \le 2$, we can get a $P_4$-sparse graph after having added the fill edges $v u$ and $v w$ (this implies that $v$ becomes universal in $(V(H) \setminus \{v\}) \cup \{w\}$) and those connecting $u$ to the vertices in $R_H$ if $R_H$ is non-empty; then the vertices in $(V(H) \setminus \{v\}) \cup \{w\}$ induce a thin spider with $K = (K_H \setminus \{v\}) \cup \{u\}$, $S = (S_H \setminus \{u\}) \cup \{w\}$, and $R = R_H$, for a total of $|K_H| - 1 + |R_H|$ fill edges (excluding the tail~$u w$). If $|R_H| \ge 2$,
a $P_4$-sparse graph is obtained after in addition to the tail~$u w$, we add the fill edges $v u$, $v w$ (again $v$ is universal in $(V(H) \setminus \{v\}) \cup \{w\}$), and the fill edges connecting $w$ to the vertices in $K_H \setminus \{v\}$ (then the vertices in $(V(H) \setminus \{v\}) \cup \{w\}$ induce a thin spider with $K = K_H \setminus \{v\}$, $S = S_H \setminus \{u\}$, and $R = R_H \cup \{u,w\}$), for a total of $|K_H| + 1$ fill edges (excluding $u w$).

Now, consider the case that $|K_H| \ge 4$. If $|R_H| \le 1$, we get a $P_4$-sparse graph after having made $u$ universal by connecting it to the remaining vertices in $S_H$ by using $|K_H| - 1$ fill edges, and adding the fill edge~$u v$, and those connecting $u$ to the vertices in $R_H$ if $R_H$ is non-empty, for a total of $|K_H| + |R_H|$ fill edges (excluding $u w$).
If $|R_H| \ge 1$,
a $P_4$-sparse graph is obtained after having made $v$ universal (by adding the fill edges~$v u$ and $v w$) and after having connected $w$ to all vertices in $K_H \setminus \{v\}$ (then the vertices in $(V(H) \setminus \{v\}) \cup \{w\}$ induce a thick spider with $K = K_H \setminus \{v\}$, $S = S_H \setminus \{u\}$, and $R = R_H \cup \{u,w\}$) for a total of $|K_H| + 1$ fill edges (excluding $u w$).

Below we show the minimality of this solution.
Recall that $v \in K_H$ is the non-neighbor of $u$ in $H$. We consider each of the five cases.
\begin{itemize}
\item[(i)]
\textit{$|K_H| = 3$ and $R_H = \emptyset$}:
Suppose, for contradiction, that there is a $P_4$-sparse completion of $H+uw$ with at most $|K_H|-2 = 1$ fill edge (excluding $u w$). If $v$ is incident on the unique fill edge (which connects $v$ to $u$ or $w$), then the vertices in $S \cup \{v,w\}$ induce an $F_3$. Now suppose that the fill edge is not incident on $v$. Moreover, there exists at least one vertex $s \in S_H \setminus \{u\}$ that is not incident on the fill edge either. Then, the vertices $u, v, w, s, k$ (where $k \in K_H$ is the non-neighbor of $s$ in $H$) induce an $F_5$ if $k, w$ are connected by the fill edge, or an $F_2$ otherwise.

\item[(ii)]
\textit{$|K_H| = 3$ and $|R_H| = 1$}:
Let $R_H = \{r\}$. Suppose, for contradiction, that there is a $P_4$-sparse completion of $H+uw$ with at most $|K_H|-1 = 2$ fill edges (excluding $u w$). We distinguish three cases depending on whether $v$ is incident on $0$, $1$, or $2$ fill edges:
\begin{itemize}
\item[$\bullet$] \textit{$v$ is not incident on a fill edge}:
\ If there exists a pair~$s, k$ of non-neighbors with $s \in S_H \setminus \{u\}$ and $k \in K_H \setminus \{v\}$ such that none of $s, k$ is incident on a fill edge to $u$ or $w$, the vertices $u, v, w, s, k$ induce an $F_2$. Otherwise, since the number of such pairs is $2$, for each such pair~$s, k$, exactly one of $s.k$ is incident on a fill edge to $u$ or $w$, and no other fill edges exist. If there exists a vertex $k \in K_H \setminus \{v\}$ not incident on a fill edge to $w$, the vertices $u, v, w, k, r$ induce an $F_5$, otherwise each of the fill edges connects each of the vertices in $K_H \setminus \{v\}$ to $w$, and then $u, v, w, s, k$ (for any pair~$s, k$ of non-neighbors with $s \in S_H \setminus \{u\}$ and $k \in K_H \setminus \{v\}$) induce an $F_5$.
\item[$\bullet$] \textit{$v$ is incident on $1$ fill edge (to $u$ or $w$)}:
\ Then, there is $1$ more fill edge; hence, there exist $2$ vertices in the set~$(S_H \setminus \{u\}) \cup \{r\}$ that are not incident on a fill edge connecting them to $u$ or $w$, and let these vertices be $p_1, p_2$. Then, the vertices $u, v, w, p_1, p_2$ induce an $F_5$ if $p_1, p_2$ are connected by a fill edge or an $F_3$ otherwise.
\item[$\bullet$] \textit{$v$ is incident on $2$ fill edges connecting it to $u$ and $w$}:
\ Then, there is no other fill edge. Then, the vertices $u, w, k, k', r$ (where $\{k, k'\} = K_H \setminus \{v\}$) induce an $F_6$.
\end{itemize}

\item[(iii)]
\textit{$|K_H| = 3$ and $|R_H| \ge 2$}:
Let $r_1, r_2$ be two vertices in $R_H$. Suppose, for contradiction, that there is a $P_4$-sparse completion of $H+uw$ with at most $|K_H| = 3$ fill edges (excluding $u w$). Again, we distinguish three cases depending on whether $v$ is incident on $0$, $1$, or $2$ fill edges:
\begin{itemize}
\item[$\bullet$] \textit{$v$ is not incident on a fill edge}:
\ Consider the case that there exists a vertex~$k \in K_H \setminus \{v\}$ that is not incident on a fill edge to $w$. Let $s \in S_H$ be the non-neighbor of $k$ in $H$ and $A = (S_H \setminus \{u,s\}) \cup \{r_1, r_2\}$; the set~$A$ contains $3$ vertices which are common neighbors of $v, k$. If at least one of these $3$ vertices (say, $p$) is not incident on a fill edge to $u, w$, then the vertices $u, v, w, k, p$ induce an $F_5$, otherwise all $3$ of these vertices are incident on a fill edge to $u, w$ (then these are all the fill edges) and the vertices $u, v, w, s, k$ induce an $F_2$. On the other hand, if no such vertex~$k$ exists, then both vertices in $K_H \setminus \{v\}$ are incident on a fill edge to $w$, accounting for $2$ of the $3$ fill edges; then there exists a vertex~$s' \in S_H \setminus \{u\}$ which is not incident on a fill edge to $w$, and the vertices $u, v, w, s', k'$ (where $k' \in K_H$ is the non-neighbor of $s'$) induce an $F_5$.
\item[$\bullet$] \textit{$v$ is incident on $1$ fill edge (to $u$ or $w$)}:
\ There are $2$ more fill edges; hence, there exist $2$ vertices in the set~$(S_H \setminus \{u\}) \cup \{r_1, r_2\}$ that are not incident on a fill edge connecting them to $u$ or $w$, and let these vertices be $p_1, p_2$. Then, the vertices $u, v, w, p_1, p_2$ induce an $F_5$ if $p_1, p_2$ are connected by a fill edge or an $F_3$ otherwise.
\item[$\bullet$] \textit{$v$ is incident on $2$ fill edges connecting it to $u$ and $w$}:
\ Then, there is $1$ more fill edge; hence, there exists a vertex $k \in K_H \setminus \{v\}$ that is not incident on the fill edge. Moreover, there exist $2$ vertices in the set~$(S_H \setminus \{u,s\}) \cup \{r_1,r_2\}$ that are not incident on a fill edge connecting them to $u$ or $w$ (where $s \in S_H$ is the non-neighbor of $k$); let these vertices be $p_1, p_2$. Then, the vertices $u, w, k, p_1, p_2$ induce an $F_5$ if $p_1, p_2$ are adjacent or an $F_3$ otherwise.
\end{itemize}

\item[(iv)]
\textit{$|K_H| \ge 4$ and $R_H = \emptyset$}:
Suppose, for contradiction, that there is a $P_4$-sparse completion of $H+uw$ with at most $|K_H|-1$ fill edges (excluding the tail~$u w$). Again, we distinguish three cases depending on whether $v$ is incident on $0$, $1$, or $2$ fill edges:
\begin{itemize}
\item[$\bullet$] \textit{$v$ is not incident on a fill edge}:
\ If there exists a vertex $s \in S_H \setminus \{u\}$ not incident on a fill edge to $u$, $w$ or to its non-neighbor~$k \in K_H$ in $H$, the vertices $u, v, w, s, k$ induce an $F_5$ if $k, w$ are connected by a fill edge, or an $F_2$ otherwise; if all vertices in $S_H \setminus \{u\}$ are incident on a fill edge to $u$, $w$, or their non-neighbor in $K_H$, then there are no more fill edges and the vertices $u, v, w, k, k'$ (for any $k, k' \in K_H \setminus \{v\}$) induce an $F_6$.
\item[$\bullet$] \textit{$v$ is incident on $1$ fill edge (to $u$ or $w$)}:
\ Then, the remaining fill edges are at most $|K_H| - 2$ in total. If there exist two vertices $s_1, s_2 \in S_H \setminus \{u\}$ not incident on a fill edge to $u$ or $w$, the vertices $u, v, w, s_1, s_2$ induce an $F_5$ or an $F_3$ depending on whether $s_1, s_2$ are connected by a fill edge or not. Thus, there cannot be two such vertices $s_1, s_2$; this implies that the remaining fill edges are precisely $|K_H| - 2$, and they connect all but one vertex in $S_H \setminus \{u\}$ to $u$ or $w$; let that vertex be $s$. Then, the vertices $u, v, w, s, k'$ (where $k' \in K_H \setminus \{v\}$ is a neighbor of $s$ in $H$) induce an $F_6$ or an $F_1$ if the fill edge incident on $v$ connects it to $u$ or $w$ respectively.
\item[$\bullet$] \textit{$v$ is incident on $2$ fill edges connecting it to $u$ and $w$}:
\ Then, the remaining fill edges are at most $|K_H|-3$ in total; hence, there exist two pairs of non-adjacent vertices $s_1,k_1$ and $s_2,k_2$ with $s_1,s_2 \in S_H \setminus \{u\}$ and $k_1, k_2 \in K_H \setminus \{v\}$ such that none of $s_1, s_2, k_1, k_2$ is incident on a fill edge to $u$ or $w$. Let $A = S_H \setminus \{u, s_1, s_2\}$; the set~$A$ is the set of $|K_H| - 3$ common neighbors of $k_1, k_2$ in $S_H$ other than $u$. If there exists a vertex $s \in A$ not incident on a fill edge to $u$ or $w$, then the vertices $u, w, k_1, k_2, s$ induce an $F_6$, otherwise, the remaining fill edges are precisely $|K_H|-3$ and they connect each of the vertices in $A$ to $u$ or $w$, that is, none of the vertices in $K_H \setminus \{v\}$ is incident on a fill edge. Then, the vertices $u, w, s_1, s_2, k$ (where $k$ is any vertex in $K_H \setminus \{v,k_1,k_2\}$) induce an $F_3$.
\end{itemize}

\item[(v)]
\textit{$|K_H| \ge 4$ and $|R_H| \ge 1$}:
Let $r \in R_H$. Suppose, for contradiction, that there is a $P_4$-sparse completion of $H + uw$ with at most $|K_H|$ fill edge (excluding the tail~$u w$). Again, w distinguish three cases depending on whether $v$ is incident on $0$, $1$, or $2$ fill edges:
\begin{itemize}
\item[$\bullet$] \textit{$v$ is not incident on a fill edge}:
\ If there exists a vertex $s \in S_H \setminus \{u\}$ not incident on a fill edge to $u$, $w$, or to its non-neighbor~$k \in K_H$ in $H$, the vertices $u, v, w, s, k$ induce an $F_5$ if $k, w$ are connected by a fill edge, or an $F_2$ otherwise; if all vertices in $S_H \setminus \{u\}$ are incident on a fill edge to $u$, $w$, or their non-neighbor in $K_H$, which account for the $|K_H| - 1$ of the $|K_H|$ fill edges, there exist vertices $k, k' \in K_H  \setminus \{v\}$ which are not incident on a fill edge and then the vertices $u, v, w, k, k'$ induce an $F_6$.
\item[$\bullet$] \textit{$v$ is incident on $1$ fill edge (to $u$ or $w$)}:
\ Then, the remaining fill edges are at most $|K_H| - 1$ in total. If all vertices in $K_H \setminus \{v\}$ are incident on a fill edge to $w$, then no more fill edges exist and the vertices $u, v, w, s_1, s_2$ (for any $s_1, s_2 \in S_H \setminus \{u\}$) induce an $F_3$. Thus, there exists $k \in K_H \setminus \{v\}$ which is not incident on a fill edge to $w$. The number of common neighbors of $v, k$ in $S_H \cup {r}$ is $|K_H| - 1$. If each of these vertices is incident on a fill edge to $u$ or $w$, then no more fill edges exist and the vertices $u, v, w, s, k'$ induce an $F_6$ or an $F_1$ depending on whether the fill edge incident on $v$ connects it to $u$ or $w$, respectively, where $s \in S_H$ is the non-neighbor of $k$ and $k'$ is any vertex in $K_H \setminus \{v, k\}$; hence, there exists a common neighbor~$p$ not incident on a fill edge to $u$ or $w$ and the vertices $u, v, w, k, p$ induce an $F_6$ or an $F_1$ depending on whether the fill edge incident on $v$ connects it to $u$ or $w$, respectively.
\item[$\bullet$] \textit{$v$ is incident on $2$ fill edges connecting it to $u$ and $w$}:
\ Then, the remaining fill edges are at most $|K_H| - 2$ in total; hence, there exists a pair of non-adjacent vertices~$s, k$ (where $s \in S_H \setminus \{u\}$ and $k \in K_H \setminus \{v\}$) which are not incident on a fill edge to $u$ or $w$. Let $A = (S_H \setminus \{u, s\}) \cup \{r\}$; the set~$A$ is a set of $|K_H| - 1$ neighbors of $k$ other than $u$. Then, there exists a vertex~$p_1$ in $A$ which is not incident on a fill edge to $u$ or $w$. If there exists a second vertex~$p_2$ in $A$ not incident on a fill edge to $u$ or $w$, then the vertices $u, w, k, p_1, p_2$ induce an $F_5$ if $p_1, p_2$ are connected by a fill edge or an $F_3$ otherwise. If each vertex in $A \setminus \{p_1\}$ is incident on a fill edge to $u$ or $w$, then the fill edges incident on these vertices account for the remaining $|K_H| - 2$ fill edges and the vertices $u, w, s, k_1, k_2$ (for any vertices $k_1, k_2 \in K_H \setminus \{v, k\}$) induce an $F_6$.
\end{itemize}

\end{itemize}
Therefore, if we use fewer than the stated number of fill edges, in each case, the resulting graph contains an induced forbidden subgraph, a contradiction.

\smallskip
2. Let $v \in S_H$ be the non-neighbor of $u$ in $H$. Then, we get a $P_4$-sparse graph by connecting $u$ to $v$; thus, $u$ becomes universal in $V(H) \cup \{w\}$. This is the minimum number of fill edges (excluding the tail~$u w$) that need to be added since for any pair of non-neighbors $s, k$ with $s \in S_H \setminus \{ v \}$ and $k \in K_H \setminus \{ u \}$, the vertices $u, v, w, s, k$ induce a forbidden subgraph~$F_3$, a contradiction.

\smallskip
3. By connecting $w$ to all vertices in $K_H$ and then computing a minimum $P_4$-sparse completion of $H[R_H \cup \{w\}]$, we get a $P_4$-sparse graph and the number of fill edges needed is $|K_H| + f'$.

To prove the minimality of this number of fill edges, suppose, for contradiction, that we can get a $P_4$-sparse graph from $H+uw$ after having added at most $|K_H| - 1$ fill edges incident on vertices in $S_H \cup K_H$ (excluding the tail~$u w$). Then, there exists a pair~$s_1, k_1$ of non-neighbors in $H$ with $s_1 \in S_H$ and $k_1 \in K_H$ none of which is incident on a fill edge to $u$ or $w$. 
We distinguish the following two cases that cover all possibilities.
\begin{itemize}
\item
\emph{Each of the vertices in $K_H \setminus \{k_1\}$ is incident on a fill edge to $w$}. These are precisely all the $|K_H| - 1$ fill edges; hence none of the vertices in $S_H \setminus \{s_1\}$ is incident on a fill edge. Then, the vertices $u, w, k_1, s_2, s_3$ (for any $s_2, s_3 \in S_H \setminus \{s_1\}$) induce an $F_3$.
\item
\emph{There exists at least one vertex in $K_H \setminus \{ k_1 \}$ that is not incident on a fill edge to $w$}. Let that vertex be $k_2$. Then, if there exists another vertex~$k_3 \in K_H \setminus \{k_1, k_2\}$ that is not incident on a fill edge to $w$ as well, the vertices $u, w, k_2, k_3, s_1$ induce an $F_6$. On the other hand, if each of the vertices in $K_H \setminus \{k_1, k_2\}$ is incident on a fill edge to $w$ (which implies that $k_3$ is adjacent to $w$), then these fill edges are $|K_H| - 2$ in total, with only $1$ remaining. If the non-neighbor~$s_3$ of $k_3$ in $S_H$ is not incident on a fill edge to $u$ or $w$, then the vertices $u, w, k_1, k_2, s_3$ induce an $F_6$ whereas if it is adjacent to $u$ or $w$, then there are no more fill edges. In particular, if $s_3$ is adjacent to $u$, the vertices $u, k_1, k_3, s_1, s_3$ induce an $F_6$ and if it is adjacent to $w$, the vertices $u, w, k_2, s_1, s_3$ induce an $F_4$.
\end{itemize}
In each case, we get a contradiction. Thus every minimum $P_4$-sparse completion of $H+uw$ requires at least $|K_H|$ fill edges incident on vertices of $S_H \cup K_H$. Now, if there exists a minimum $P_4$-sparse completion~$H'$ of $H+uw$ having fewer than $|K_H| +f'$ fill edges, then the fact that at least $|K_H|$ of them are incident on vertices in $S_H \cup K_H$ implies that $H'[R_H \cup \{w\}]$ is $P_4$-sparse using fewer than $f'$ fill edges in contradiction to the minimality of $f'$.
\end{proof}

If the (thin or thick) spider~$H$ belongs to a more general $P_4$-sparse graph, then Lemmas \ref{lemma:tail_thin} and \ref{lemma:tail_thick} imply the following result.

\begin{corollary} \label{corol:special}
Let $u$ be a vertex of a $P_4$-sparse graph to which we add the tail~$uw$. Let $t_0 \cdots t_h u$ be the path in the $P_4$-sparse tree of $G$ from the the root~$t_0$ to the leaf for $u$ and let $V_0, \ldots, V_h$ be the corresponding vertex sets as mentioned before. Then, if the parent~$t_h$ of $u$ is a $2$-node corresponding to a spider~$H$, the number of fill edges needed for a minimum $P_4$-sparse completion of the graph~$G+uw$ (excluding the tail~$uw$) does not exceed the minimum between
\begin{itemize}
\item[(i)]
the minimum number given by Lemmas \ref{lemma:tail_thin} and \ref{lemma:tail_thick} (if $H$ is thin or thick, respectively) augmented by $|N_G(u) \cap (V_0 \cup \cdots \cup V_{h-1})|$ \qquad and
\item[(ii)]
$\min_{t_i = 1\text{- or } 2\text{-node}} \{\, |N_G(u) \cap (V_0 \cup \ldots \cup V_{i-1})| + |(V_i \cup \ldots \cup V_h) \setminus N_G(u)| \,\}$.
\end{itemize}
\end{corollary}
\noindent
Case~(i) of Corollary~\ref{corol:special} corresponds to doing a minimum $P_4$-completion of the graph $H+uw$ and not changing the rest of the $P_4$-sparse tree~$T_G$ of $G$ whereas Case~(ii) corresponds to making $u$ universal in $H$ and then trying Formation~$2$ above each $1$-node or $2$-node in the path~$t_0 \cdots t_h$ of $T_G$.

\subsection{The Algorithm}

Recall that $t_0 t_1 \cdots t_h u$ is the path in the $P_4$-sparse tree~$T_G$ of $G$ from the root~$t_0$ to the leaf for $u$, and $V_i$ ($0 \le i < h$) is the set of vertices associated with the leaves of the subtrees rooted at the children of $t_i$ except for $t_{i+1}$ and $V_h$ is the set of vertices associated with the leaves of the subtrees rooted at the children of $t_h$ except for the leaf corresponding to $u$. See Figure~\ref{fig:labelV}.

Next we prove the conditions under which a minimum $P_4$-sparse completion of the graph~$G+uw$ uses fewer fill edges than when using Formation~$1$ or $2$.

\begin{lemma} \label{lemma:p4_formation}
There exists a minimum $P_4$-sparse completion~$G_{OPT}$ of the graph $G+uw$ which uses fewer fill edges than when using Formation $1$ or $2$ if and only if $uw$ is a wing of a $P_4$ in $G_{OPT}$ which implies that
(i)~either $u$ is a vertex of a spider in $G$ (Lemmas \ref{lemma:tail_thin} and \ref{lemma:tail_thick} apply)
(ii)~or there exists $j$ ($0 \le j < h$) such that $t_j$ is a $1$-node, $t_{j+1}$ is a $0$-node, and there exist vertices $a,b$ such that $a \in V_j$ is universal in $G[V_j]$ and $b \in V_{j+1}$ is isolated in $G[V_{j+1}]$.\\
Then, in $G_{OPT}$, the vertices $u,w,a,b$ induce a $P_4$ in a spider $(S,K,R)$ with $S = \{w,b\}$, $K = \{u,a\}$ and $R = (V_{j+1} \setminus \{b\}) \cup V_{j+2} \cup \cdots \cup V_h$.
\end{lemma}

\begin{proof}
If Formation~$1$ or Formation~$2$ is not to be used then Lemma~\ref{lemma:forms} implies that in $G_{OPT}$, $uw$ is the wing of a $P_4$. If $u$ is a vertex of a spider, then Lemmas \ref{lemma:tail_thin} and \ref{lemma:tail_thick} apply. So, in the following, assume that $u$ is not a vertex of a spider.

For the tail~$uw$ to be the wing of a $P_4$ in $G_{OPT}$, we can show that there exist vertices $x,y$ such that $u x y$ is a $P_3$ in the graph~$G$: if $u,x,y$ do not all belong to the same connected component of $G$, then we could add the tail~$uw$ to the connected component of $G$ to which $u$ belongs, thus using fewer fill edges than in $G_{OPT}$, a contradiction; if $u,x,y$ belong to the same connected component of $G$ but do not form a $P_3$, then because $u,y$ are not adjacent in $G_{OPT}$ and thus neither in $G$, $u,y$ are at distance~$2$ in $G$ and there exists a $P_3$~$uay$ in $G$ (note that $u,y$ cannot be at distance~$\ge 4$ in $G$ since then $G$ would contain an induced $P_5 = F_2$, and they cannot be at distance~$3$ either since then there exists a $P_4$~$u a b y$ in $G$ and $u$ would be a vertex of a spider in $G$).

Therefore, in the following, consider that the minimum $P_4$-sparse completion~$G_{OPT}$ of $G+uw$ contains an induced $P_4$~$wuab$ such that the graph~$G$ contains the induced $P_3$~$uab$. 
So, since $u,b$ are not adjacent in $G_{OPT}$, then they are not adjacent in $G$ either, and thus their least common ancestor~$t_k$ in the $P_4$-sparse tree~$T_G$ of $G$ is a $0$-node; it cannot be a $2$-node since then $u$ would be a vertex of a spider. Moreover, $a$ is a common neighbor of both $u,b$ and thus the least common ancestor~$t_j$ of $a,u$ in $T_G$ is a $1$- or a $2$-node (in the latter case, $a$ is a vertex of the clique of the spider) and $j < k$.

Let us now try forming the $P_4$~$wuab$, which clearly will belong to a spider, say $W = (S_W, K_W, R_W)$.
We show that $|S_W| = |K_W| = 2$. First, note that the edge~$a b$ cannot belong to a spider in $G$, since then $u$ would belong to that spider as well (note that the vertices of $G$ not belonging to a spider are either adjacent to all vertices of the spider or to none of them),
So, suppose for contradiction that the spider~$W$ has $|S_W| = |K_W| \ge 3$ and let $w, b, s \in S_W$ and $u, a, k \in K_W$ with the corresponding $S$-$K$ pairs being $w$ and $u$, $b$ and $a$, and $s$ and $k$. The spider~$W$ can be thin or thick.
\begin{itemize}
\item
\emph{The spider~$W$ is thin}.
\ Then, $b a \in E(G)$ otherwise the removal of $b a$ would produce a $P_4$-sparse graph with fewer fill edges ($b$ is isolated in $G[V(W)]$), a contradiction; similarly, $s k \in E(G)$. Moreover, $a k \in E(G)$: as above, if $a, k$ do not belong to the same connected component of the induced subgraph~$G[V(W)]$, then by adding the tail~$uw$ to the connected component of $G[V(W)]$ to which $u$ belongs would result into fewer fill edges; if $a, k$ belong to the same connected component of $G[V(W)]$ then there exists a chordless path~$\rho$ connecting them in the subgraph~$G[K_W \cup R_W]$ and the vertices in $V(\rho) \cup \{b,s\}$ induce a $P_\ell$ with $\ell \ge 5$, in contradiction to the $P_4$-sparseness of $G$. But then, $G$ contains the $P_4$~$b a k s$ and $a b$ belongs to a spider.
\item
\emph{The spider~$W$ is thick}.
\ Then, $w \in S_W$ is incident on the tail~$uw$ and $|K_W|-2 \ge 1$ fill edges. Since we can make $u$ universal in $G[V(W) \setminus \{w\}]$ by using a single fill edge and then use Formation~$2$, it is clear that building the spider~$W$ does not result into fewer fill edges.
\end{itemize}
Thus, $G_{OPT}$ with a spider~$W$ with $|K_W| \ge 3$ has no fewer fill edges than if we use Formation~$2$. Therefore, the $P_4$~$wuab$ belongs to a spider with clique size equal to $2$, which thus is thin.
Then, Property~P1 in Lemma~\ref{lemma:p4sparse_tree} implies that $w$, $u$, and $a$ are adjacent to all the neighbors of $b$ except for $a$ in $G_{OPT}$ and thus at least to the neighbors of $b$ in $G$; thus, in $G_{OPT}$,
\begin{itemize}
\item
fill edges connect vertex~$w$ to the vertices in $((V_0 \cup \cdots \cup V_{k-1}) \setminus \{a\}) \cap N_G(b)  = [(V_0 \cup \cdots \cup V_{k-1}) \setminus \{a\}] \cap N_G(u)$;
\item
vertex~$u$ and $w$ are adjacent to all neighbors of $b$ in $V_k$, that is, to the vertices in $(V_k \cap N_G(b)) \setminus N_G(u)$;
\item
vertex~$a$ is adjacent to all the vertices in $(V_j \cap N_G(b))$ and thus fill edges connect $a$ to all vertices in $(V_j \cap N_G(b)) \setminus N_G[a]$ $= (V_j \cap N_G(u)) \setminus N_G[a]$.
\end{itemize}
Additionally, Property~P2 in Lemma~\ref{lemma:p4sparse_tree} implies that because $a$ is adjacent to all the vertices in $V_{j+1} \cup \cdots \cup V_h$ and to the vertices in $V_j \cap N_G(a)$ in $G_{OPT}$ (because it is adjacent to them in $G$), then so must be vertex~$u$ in $G_{OPT}$; thus, in $G_{OPT}$, fill edges connect $u$ to all the vertices in $(V_{j+1} \cup \cdots \cup V_h) \setminus N_G(u)$ and the vertices in $(V_j \cap N_G(a)) \setminus N_G(u)$ (the set~$(V_j \cap N_G(a)) \setminus N_G(u)$ is non-empty if and only if $t_j$ is a $2$-node).

Now, let us consider using Formation~$2$ right below node~$t_j$ in the $P_4$-sparse tree~$T_G$ of $G$; then, the number of fill edges is $|(V_{j+1} \cup \cdots \cup V_h) \setminus N_G(u)| + |(V_0 \cup \cdots \cup V_{j}) \cap N_G(u)|$; the former term corresponds to fill edges incident on $u$, the latter to fill edges incident on $w$.
Then, because $j < k$ and $|((V_0 \cup \cdots \cup V_{k-1}) \setminus \{a\}) \cap N_G(u)| = |(V_0 \cup \cdots \cup V_{k-1}) \cap N_G(u)| - 1$, the only possibility for $G_{OPT}$ to use fewer fill edges than using Formation~$2$ after node~$t_j$ requires that
\begin{enumerate}
\item
$k = j+1$;
\item
$(V_j \cap N_G(u)) \setminus N_G[a] = \emptyset$;
\item
$(V_k \cap N_G(b)) \setminus N_G(u) = \emptyset$ which implies that $b$ is isolated in $G[V_k]$;
\item
$(V_j \cap N_G(a)) \setminus N_G(u) = \emptyset$ which implies that $t_j$ is a $1$-node.
\end{enumerate}
Requirement~$4$ implies that $V_j \cap N_G(u) = V_j$ which together with Requirement~$2$ imply that $N_G[a] = V_j$, that is, $a$ is universal in $G[V_j]$, and we have the second case in the statement of the lemma.
\end{proof}

Now we are ready to describe our algorithm for counting the number of fill edges in a minimum $P_4$-sparse completion of the graph~$G + uw$.


\bigskip\noindent
\textbf{Algorithm $P_4$-sparse-Tail-Addition}\\
\emph{Input}: a $P_4$-sparse graph~$G$, a vertex~$u \in V(G)$, and a tail~$uw$ to be added to $G$.\\
\emph{Output}: the minimum number of fill edges (excluding the tail~$uw$) needed in a\\
\phantom{\emph{Output}:} $P_4$-sparse completion of the graph~$G+uw$.

\smallskip\noindent
\textbf{if}\  $|V(G)| = 1$ \ \textbf{then} \quad $\{$\textit{$V(G) = \{u\} \Longrightarrow$ the graph~$G+uw$ is $P_4$-sparse}$\}$\\
\phantom{xx} \textbf{return}(0);

\smallskip\noindent
compute the path~$t_0 t_1 ... t_h$ ($h \ge 1$) from the root $t_0$ of the $P_4$-sparse tree of $G$ to the parent-node~$t_h$ of the leaf corresponding to $u$;\\
compute the sets of vertices $V_i$. $0 \le i \le h$ (see Figure~\ref{fig:labelV});

\smallskip\noindent
$min \gets |N_G(u)|$; \quad $\{$\textit{corresponds to Formation~$1$}$\}$

\smallskip\noindent
$\{$\textit{check for Formation~$2$ (Lemma~\ref{lemma:forms}(iii) and Case~(ii) of Corollary~\ref{corol:special})}$\}$\\
\textbf{for}\  each $t_i$ ($i=0,1,\ldots,h$) that is a $1$- or a $2$-node \ \textbf{do}\\
\phantom{xx} $\ell \gets |N_G(u) \cap (V_0 \cup \cdots \cup V_{i-1})| + |(V_i \cup \cdots \cup V_h) \setminus N_G(u)|$;\\
\phantom{xx} update $min$ if $\ell < min$;\\

\smallskip\noindent
$\{$\textit{check for new $P_4$ formation (Lemma~\ref{lemma:p4_formation})}$\}$\\
\textbf{for}\  each $i=0,1,\ldots,h-1$ such that $t_i$ is a $1$-node and $t_{i+1}$ is a $0$-node \ \textbf{do}\\
\phantom{xx} \textbf{if}\  there exist vertex~$a \in V_i$ such that $a$ is universal in $V_i$
\ \textbf{and}\\
\phantom{xx \textbf{if}} 
\phantom{there exist} vertex~$b\in V_{i+1}$ such that $b$ has no neighbors in $V_{i+1}$ \ \textbf{then}\\
\phantom{xx} $\ell \gets |N_G(u) \cap (V_0 \cup \cdots \cup V_{i-1})| + |V_i \setminus \{a\}| + |V_{i+1} \setminus \{b\}| +\\
\phantom{xxxxxxxxxxxxxxxxx}
\phantom{xxxxxxxxxxxxxxxxx} |(V_{i+2} \cup \cdots \cup V_h) \setminus N_G(u)|$;\\
\phantom{xx} update $min$ if $\ell < min$;\\

\smallskip\noindent
$\{$\textit{check the cases if $t_h$ is a $2$-node and apply case~(i) of Corollary~\ref{corol:special}}$\}$\\
\textbf{if} \ $t_h$ is a $2$-node \ \textbf{then}\\
\phantom{xx} $\ell \gets$ number of fill edges according to the cases of Lemmas \ref{lemma:tail_thin} or \ref{lemma:tail_thick};\\
\phantom{xx} $\ell \gets \ell + |N_G(u) \cap (V_0 \cup \cdots \cup V_{h-1})|$;
\qquad $\{$\textit{Case~(i) of Corollary~\ref{corol:special}}$\}$\\
\phantom{xx} update $min$ if $\ell < min$;\\
\textbf{return}($min$);

\bigskip
Algorithm $P_4$-sparse-Tail-Addition can be easily augmented to return a minimum cardinality set of fill edges. 
The correctness of the algorithm follows from Lemmas \ref{lemma:forms}, \ref{lemma:tail_thin}, \ref{lemma:tail_thick}, \ref{lemma:p4_formation}, and Corollary~\ref{corol:special}.
Let $G$ be the given graph and let $n$ be the number of its vertices. If the $P_4$-sparse tree~$T_G$ of $G$ is given, an $O(n)$-time traversal of the tree enables us to compute the path $t_0 t_1 \cdots t_h u$ and the number of neighbors and non-neighbors of $u$ in each of the sets $V_0, \ldots, V_h$; additionally, the height of $T_G$ is $O(n)$ and thus $h = O(n)$. Since the conditions of Lemmas \ref{lemma:tail_thin} and \ref{lemma:tail_thick} can be checked in $O(1)$-time, the entire algorithm runs in $O(n)$ time.

\begin{theorem}
Let $G$ be a $P_4$-sparse graph on $n$ vertices and let $uw$ be tail attached at node~$u$ of $G$. If the $P_4$-sparse tree of $G$ is given, Algorithm $P_4$-sparse-Tail-Addition computes the minimum number of edges to be added to $G+uw$ so that the resulting graph is $P_4$-sparse in $O(n)$ time.
\end{theorem}

If the $P_4$-sparse tree~$T_G$ of $G$ is not given, then it can be computed in $O(n+m)$ time where $m$ is the number of edges of $G$ \cite{JO92a}, and the entire algorithm takes $O(n+m)$ time.


\end{document}